\documentclass[letterpaper, 10 pt, conference]{ieeeconf}  

\IEEEoverridecommandlockouts                              

\newtheorem{mydef}{\bf Definition}
\newtheorem{mythm}{\bf Theorem}
\newtheorem{myprob}{\bf Problem}
\newtheorem{mylem}{\bf Lemma}

\newtheorem{mypro}{\bf Proposition}
\newtheorem{myexm}{\bf Example}
\newtheorem{remark}{Remark}
\usepackage{amsmath}
\usepackage{tikz-cd}
\usepackage{amssymb}
\usepackage{graphicx}
\usepackage{dblfloatfix}
\usepackage{hyperref}
\usepackage{ifthen}
\usepackage{color}
\usepackage{booktabs}
\usepackage{multirow}

\usepackage{cite} 
\usepackage{multicol}
\usepackage{caption} 
\usepackage{subfigure}
\usepackage{mathrsfs}

\usepackage[algo2e,ruled,linesnumbered,vlined,boxed,commentsnumbered]{algorithm2e}
\usepackage{algorithmic, algorithm}

\title{\LARGE \bf
On the Computation and Approximation of Backward Reachable Sets  \\
for Max-Plus Linear Systems using Polyhedras
}

\author{Yuda Li, Shaoyuan Li and Xiang Yin%
\thanks{This work was supported by  the National Natural Science Foundation of China (62061136004, 62173226, 61833012).}
	\thanks{Yuda Li, Shaoyuan Li and Xiang Yin are with Department of Automation and Key Laboratory of System Control and Information Processing, Shanghai Jiao Tong University, Shanghai 200240, China.
	{\tt\small \{yuda.li, syli, yinxiang\}@sjtu.edu.cn}.}
}

\begin{document}

\renewcommand{\algorithmicrequire}{\textbf{Input:}}
\renewcommand{\algorithmicensure}{\textbf{Output:}}
\def\XY#1{{\textcolor{red}{ {\bf XY:} #1}}}
\def\CY#1{{\textcolor{blue}{ {\bf CY:} #1}}}
\def\zn#1{{\textcolor{blue}{#1}}}
\newcommand{\trash}[1]{%
    \stepcounter{Trash} 
    \textcolor{blue}{\#\theWSQcomment\ \textbf{#1}} 
}
\maketitle
\thispagestyle{empty}
\pagestyle{empty}

\begin{abstract}
This paper investigates reachability analysis for max-plus linear systems (MPLS), an important class of dynamical systems that model synchronization and delay phenomena in timed discrete-event systems.
We specifically focus on backward reachability analysis, i.e., determining the set of states that can reach a given target set within a certain number of steps. Computing backward reachable sets presents significant challenges due to the non-convexity of max-plus dynamics and the complexity of set complement operations.
To address these challenges, we propose a novel approximation framework that efficiently computes backward reachable sets by exploiting the structure of tropical polyhedra. Our approach reformulates the problem as a sequence of symbolic operations and approximates non-convex target sets through closure operations on unions of tropical polyhedra. We develop a systematic algorithm that constructs both outer (M-form) and inner (V-form) representations of the resulting sets, incorporating extremal filtering to reduce computational complexity. 
The proposed method offers a scalable alternative to traditional DBM-based approaches, enabling reliable approximate backward reachability analysis for general target regions in MPLS.
\end{abstract}

\section{Introduction}
Max-plus linear systems (MPLS) are a class of discrete-event systems that operate over the max-plus algebra, where the addition operation is replaced by the maximum and the multiplication operation is replaced by conventional addition \cite{hardouin2018control,de2020analysis}. It captures the synchronization and delay phenomena in timed discrete-event systems and is related to timed event graphs \cite{amari2011max,he2021performance}, a subclass of Petri nets where each place has exactly one input and one output transition. 
These systems are particularly useful for modeling and analyzing scheduling problems \cite{he2016cycle}, manufacturing systems \cite{hardouin2011towards,imaev2008hierarchial,van2020model}, and transportation networks \cite{kistosil2018generalized,kersbergen2016towards,outafraout2020control}, where the focus is on determining the occurrence times of events under precedence constraints. 

In the context of MPLS analysis and control, reachability analysis is one of the most fundamental problems. For instance, given the current state of the system, one may need to compute the set of states reachable at future time instants, either under a specific control sequence or all possible inputs. This is known as the \emph{forward reachability} problem, which frequently arises in applications such as model predictive control \cite{xu2016optimistic,xu2018model}.  
Conversely, some applications require computing the set of states that can reach a given target state or region within a certain number of steps. This is referred to as the \emph{backward reachability} problem and is particularly relevant in safety analysis, where identifying invariant sets is critical for ensuring system stability and performance \cite{soudjani2016formal,pratama2018safety}. Both forward and backward reachability analyses offer essential insights into system behavior, enabling the design of effective control strategies and ensuring compliance with safety constraints.

In the literature, many existing works have studied reachability analysis for MPLS. 
For example, in \cite{4099499}, the authors extended the concept of geometric invariance from conventional linear systems to the max-plus setting. They characterized invariant semimodules over an infinite horizon, providing valuable insights into the asymptotic behavior of MPLS. 
However, this approach primarily addresses infinite-horizon invariance and cannot directly compute reachable sets at finite time steps.
For finite-horizon reachability computation, several alternative methods have been proposed. For instance, \cite{adzkiya2014forward,mufid2021smt,adzkiya2013finite} leverage difference-bound matrices (DBMs) and piecewise affine systems to partition the state space into multiple regions, computing local reachable sets within each region. However, since DBMs divide the state space into exponentially many regions, the computational complexity grows prohibitively high, particularly as the system dimension increases.

More recently, the authors in \cite{allamigeon2013computing} proposed a novel characterization of a class of sets in MPLS called tropical polyhedra, demonstrating that they can be efficiently represented using a generating set of basic vectors. Building on this result, \cite{espindola2022stochastic,espindola2025set} investigated reachability set computation via tropical polyhedra. Specifically, they investigated 
(i) the forward reachability problem for polyhedra-represented state sets, and
(ii) the backward reachability problem from a single point, which is useful for state estimation within their framework. 
However, it still remains an open question how to compute backward reachability from a general set represented by polyhedra. This problem is particularly important in applications such as safety analysis, where it is necessary to propagate the safety region rather than focusing solely on a single point.

In this paper, we address the challenging problem of computing backward reachable sets in max-plus linear systems.
This problem is fundamentally more difficult than the forward reachability problem studied in \cite{espindola2025set}, as one needs to consider all possible pre-images of a target set (instead of a point) under max-plus dynamics, which often involves solving non-convex set-valued equations and handling set complements.
Specifically, we tackle this challenge by developing a novel computational framework that efficiently approximates backward reachable sets through a combination of set operations and symbolic propagation techniques. Our approach goes beyond the basic setting by allowing complement operators to represent target sets. However, in such cases, explicitly computing the backward reachable set becomes significantly more difficult, as the boundary of the set may become highly irregular or non-polyhedral.
To address this, we further propose a new technique based on tropical polyhedras, which allows us to efficiently approximate generalized reachable sets that involve complements or unions of basic sets.
In summary, our method provides a computationally efficient and scalable approach to approximately computing backward reachable sets in max-plus linear systems, extending beyond traditional polyhedral representations and supporting a broader class of target sets.

The rest of the paper is organized as follows.
In Section \ref{sec: preliminary}, we recall the definitions of max-plus linear systems, tropical cones, and polyhedra, which are necessary for later use.
In Section \ref{sec: problem formulation}, we characterize the MPL system and define the backward operator.
In Section \ref{sec: Basic ideas}, we introduce the framework of our method.
In Section \ref{sec: approximation method}, we provide the details of how to approximate sets using tropical polyhedra and prove the soundness of this approach.
In Section \ref{sec: reachability}, we demonstrate how to compute the backward reachable set using the approximation method introduced in the previous section.
In Section \ref{sec: Improvement}, we discuss how to improve the efficiency of our algorithm for set approximation. 
Finally, we conclude the paper in Section~\ref{sec:conclu}.

\section{Preliminary}\label{sec: preliminary}
In this section, we review some basic concepts in max-plus algebra that we will use in the latter developments. 
\subsection{Dioid and Max-Plus Algebra}
A dioid is a set $\mathcal{M}$ together with two binary operations $\otimes:\mathcal{M} \times \mathcal{M} \to \mathcal{M}$ and $\oplus:\mathcal{M} \times \mathcal{M} \to \mathcal{M}$ such that 
$\oplus$ is commutative and idempotent ($a \oplus a = a$), 
and 
$\otimes$ is distributive over $\oplus$. 
There are two elements $\varepsilon$ and $e$ such that $a \oplus \varepsilon = \varepsilon \oplus a = a$ and $a \otimes e = e \otimes a = a$. 
A max-plus algebra is a special kind of dioid with 
$\mathcal{D} = \mathbb{R}_{\max} = \mathbb{R} \cup \{-\infty\}$ or $\mathcal{D} = \mathbb{Z}_{\max} = \mathbb{Z} \cup \{-\infty\}$, together with two binary operations $\oplus: \mathcal{D} \times \mathcal{D} \to \mathcal{D}$ and $\otimes: \mathcal{D} \times \mathcal{D} \to \mathcal{D}$ defined by $a \oplus b = \max(a,b)$ and $a \otimes b = a + b$. 
For any $a \in \mathcal{D}$, we define $(-\infty) \oplus a = a \oplus (-\infty) = a$ and $(-\infty) \otimes a = a \otimes (-\infty) = -\infty$. For simplicity, hereafter we denote $-\infty$ as $\varepsilon$. On $\mathcal{D}$, we can define a metric given by $d(x,y) = |\exp(x) - \exp(y)|$, with the convention that $\exp(\varepsilon) = 0$. The closed sets and open sets are defined with respect to the topology generated by this metric.

\subsection{Semimodules and Sub-Semimodules on $\mathcal{D}$}
The Cartesian product of $n$ copies of $\mathcal{D}$, denoted as $\mathcal{D}^n$, can naturally be equipped with the operation 
$\oplus:\mathcal{D}^n \times \mathcal{D}^n \to \mathcal{D}^n$ of vector addition defined by $(a_1,\dots,a_n)\oplus(b_1,\dots,b_n) = (a_1\oplus b_1,\dots,a_n\oplus b_n)$, and the operation of scalar multiplication $\cdot:\mathcal{D} \times \mathcal{D}^n \to \mathcal{D}^n$ defined by $a \cdot (b_1,\dots,b_n) = (a \otimes b_1,\dots,a \otimes b_n)$. The tuple $(\mathcal{D}^n,\oplus,\cdot)$ is called a max-plus $\mathcal{D}$-semimodule or simply a semimodule. It is also equipped with a natural partial order defined by $x \preceq y \leftrightarrow \forall 1 \le i \le n,\ x_i \le y_i$. The element $\varepsilon_n = (\varepsilon,\dots,\varepsilon) \in \mathcal{D}^n$ is the minimal element. A matrix $A$ of $\mathcal{D}^{n\times m}$ is associated with the linear map $\phi_A:\mathcal{D}^m \to \mathcal{D}^n$ defined by $\phi_A(x)_i = \bigoplus_{j=1}^m A_{ij} \otimes x_j$ for $x = (x_1,\dots,x_m)$. For the rest of the paper, we denote $\phi_A(x)$ simply as $A \otimes x$. Similarly, for $a,b \in \mathcal{D}^n$, the scalar product of $a,b$ is defined as $(a|b) = \bigoplus_{i=1}^n a_i \otimes b_i$. A subsemimodule is a subset $\mathcal{S} \subseteq \mathcal{D}^n$ such that $\mu \cdot x \oplus \lambda \cdot y \in \mathcal{S}$ for all $\mu,\lambda \in \mathcal{D}$. For a vector $x \in \mathcal{D}^n$, we define the support of $x$ as $\mathrm{Supp}(x) = \{i \in [n] \mid x_i \neq \varepsilon\}$, where $[n]$ denotes the set $\{1,\dots,n\}$. For two vectors $x_1,x_2 \in \mathcal{D}^n$, we have $\mathrm{Supp}(x_1 \oplus x_2) = \mathrm{Supp}(x_1) \cup \mathrm{Supp}(x_2)$. It is thus possible to define the support of a subsemimodule: Let $S$ be a subsemimodule, we define the support of $S$, denoted as $\mathrm{Supp}(S)$, as the largest possible support of a vector in $S$, i.e., $\mathrm{Supp}(S) = \bigcup_{x \in S} \mathrm{Supp}(x)$. Similarly, the topology on $\mathcal{D}^n$ is the product topology; open sets and closed sets are defined with respect to this topology.

\subsection{Tropical Cones, Half-Spaces and Polyhedras}
A tropical half-space (resp. affine half-space) $\mathscr{H}$ is a subset of $\mathcal{D}^n$ of the form $\mathscr{H} = \{x \in \mathcal{D}^n \mid (a|x) \le (b|x)\}$ (resp. $\mathscr{H} = \{x \in \mathcal{D}^n \mid (a|x) \oplus c \le (b|x) \oplus d\}$), where $a, b \in \mathcal{D}^n$, $c, d \in \mathcal{D}$. A tropical cone (resp. a tropical polyhedra) is an intersection of finitely many half-spaces (resp. affine half-spaces), or equivalently, a set of the form $\mathscr{C} = \{x \in \mathcal{D}^n \mid A \otimes x \le B \otimes x\}$ (resp. $\mathscr{P} = \{x \in \mathcal{D}^n \mid A \otimes x \oplus c \le B \otimes x \oplus d\}$) for $A, B \in \mathcal{D}^{q \times n}$ and $c, d \in \mathcal{D}^n$, which we denote for simplicity as $\mathscr{C} = \langle A, B \rangle$ (resp. $\mathscr{P} = \langle (A, c), (B, d) \rangle$). We call this kind of representation an outer representation, or $\mathcal{M}$-form of a tropical cone (or tropical polyhedra). It has been proved by \cite{butkovivc2010max} that every tropical cone $\mathscr{C}$ can be generated by a finite set, i.e.,
\begin{equation}\label{V form of cone}
        \exists v_1, \dots, v_N \in \mathscr{C},\ \mathscr{C} = \left\{\left.\bigoplus_{i=1}^N \lambda_i \cdot v_i \right| \lambda_1, \dots, \lambda_N \in \mathcal{D}\right\}.
\end{equation}
We call $\{v_1, \dots, v_N\}$ the inner representation of $\mathscr{C}$, or its $\mathcal{V}$-form, and denote $\mathscr{C}$ as $\mathrm{Span}(\{v_1, \dots, v_N\})$. Similarly, a tropical polyhedra has an equivalent representation
\begin{equation}\label{V form of polyhedra}
    \left\{\left.\left(\bigoplus_{i=1}^N \lambda_i \cdot v_i\right) \oplus \left(\bigoplus_{j=1}^M \mu_j \cdot e_j\right) \right| \lambda_i, \mu_j \in \mathcal{D},\ \bigoplus_{j=1}^M \mu_j = 0\right\},
\end{equation}
where $v_i, e_j$ are vectors in $\mathcal{D}^n$. We denote by $(\{v_i\}_i, \{e_j\}_j)$ the inner representation ($\mathcal{V}$-form) of a tropical polyhedra. An algorithm for computing the generating set of a tropical cone in its $\mathcal{M}$-form is provided in \cite{allamigeon2013computing}. Reciprocally, given a set defined by formula \eqref{V form of cone} (resp. by formula \eqref{V form of polyhedra}), it has been shown in \cite{gaubert2011minimal} that it is also a tropical cone (resp. tropical polyhedra). A method of calculating the $\mathcal{V}$-form for the tropical cone is provided in \cite{gaubert2011minimal}.

We define some notations: for a vector $v \in \mathcal{D}^n$, we denote by $v^{(\le k)}$, $v^{(\ge k)}$, and $v^{(\ge k,\ \le m)}$ the vectors formed by the first $k$ coordinates, the last $n-k+1$ coordinates, and the coordinates with indices ranging from $k$ to $m$ of the original vector $v$, respectively. For a matrix $A$, we denote by $A^{(\le k)}$, $A^{(\ge k)}$, and $A^{(\ge k, \le m)}$ the matrices formed by the first $k$ columns, the last $n-k+1$ columns, and the columns with indices ranging from $k$ to $m$ of the original matrix $A$, respectively. We denote by $A|^{(\le k),(\le m)}$ the matrix formed by the first $k$ columns and the first $m$ rows of $A$.

\subsection{Difference Bounds Matrices}
Define the regular algebra as the set $\mathcal{R} = \mathbb{R}\cup\{-\infty, \infty\}\times \{<,\le\}$ together with a linear ordre defined by $(x_1,\sim_1) \preceq (x_2, \sim_2)$ iff $x_1 > x_2$ or $x_1 = x_2$ and $(\sim_1,\sim_2) \in \{(<,<),(\le,<),(\le,\le)\}$. It is equipped with the max-plus algebraic operator $\oplus$ defined by $r_1 \oplus r_2 = \max(r_1,r_2)$ for $r_1,r_2 \in \mathcal{R}$, and $\otimes$ is defined by $(x_1,\sim_1)\otimes (x_2,\sim_2) = (x_1+x_2,\sim_1\otimes \sim_2)$, with $\le \otimes \le = \le$, $\le \otimes < = <$, $< \otimes < = <$. The neutral element for $\oplus$ is $(\infty,\le)$, and that for $\otimes$ is $(0,\le)$
\begin{mydef}
    A difference bounds matrix (DBM) is a matrix with coefficient in $\mathcal{R}\setminus \{(-\infty,<)\}$. For the max-plus semimodule $\mathcal{D}^n$ and $M \in \mathcal{R}^{n\times n}$ a DBM, define the region associated to $M$ by
    \[\mathscr{R}(M)  = \bigcap_{1\le i,j \le n}\left\{x \in \mathcal{D}^n\left| 
    \begin{aligned}
        &x_i \sim_{i,j} m_{i,j}+x_j\\ 
        &\text{where }(m_{i,j},\sim_{i,j}) = M_{i,j} 
    \end{aligned}
    \right\}\right.\]
\end{mydef}
with the convention $\infty + (-\infty) = \infty$. Under this notation, the intersection of sets of DBM's form is easy to be expressed: $\mathscr{R}(M_1)\cap \mathscr{R}(M_2) = \mathscr{R}(M_1\oplus M_2)$. In the suite, we call for simplicity a set a DBM by meaning that it is the region associated to some DBM. We say a DBM $M$ is closed if it only has $\le$ in its expression, open if only $<$ in its expression.

The Kleene star operator $*$ is defined as $M^* = \bigoplus_{k = 0}^\infty M ^{\otimes k}$, where $M^{\otimes k}$ is the $k$-th power of the matrix $M$ under multiplication of max-plus sense, and $M^{\otimes 0}$ is the identity matrix with $(0,\le)$ on the diagonal and $(\infty,\le)$ elsewhere. We can prove that $\mathscr{R}(M) = \mathscr{R}(M^*)$, the matrix $M^*$ is called the canonical form of $M$, and the region $\mathscr{R}(M)$ is non empty iff all coefficients on the diagonal of $M^*$ equal to $(0,\le)$. Let $\mathbb{P}^{n}_k:\mathcal{D}^n \to \mathcal{D}^k$ be the projector on the first $k$ coordinates, then we have $\mathbb{P}_k^n(\mathscr{R}(M^*)) = \mathscr{R}\left(M^*|^{(\le k),(\le k)}\right)$.

There is some correlation between DBMs, tropical polyhedras and tropical cones. It can be showed that a closed DBM is a tropical cone (of the form $\langle A,\mathrm{Id} \rangle$) and every tropical cone is a finite union of closed DBMs. However, the number of DBMs needed to express a tropical cone grows exponentially in terms of the number of inequalities needed to define it. We can also set up a link between tropical cones and tropical polyhedras, as described in the following lemma:
\begin{mylem}\label{instection with Dn}
    Let $\mathcal{P} \subseteq \mathcal{D}^n$ be a tropical polyhedra, then there exists $\mathcal{C} \subseteq \mathcal{D}^{n+1}$ a tropical cone such that 
    \[\mathcal{P} = \{(x_1,\dots,x_n)\in \mathcal{D}^n\mid (0,x_1,\dots,x_n) \in \mathcal{C}\}\]
    which we denote abusively as $\mathcal{P} = \mathcal{C} \cap \mathcal{D}^n$.
\end{mylem}
\begin{proof}
    Consider the set 
    \[\mathcal{C} = \{(\lambda,x) \in \mathcal{D}\times \mathcal{D}^n \mid A\otimes x \oplus \lambda\cdot c  \preceq B \otimes x \oplus \lambda \cdot d\}\]
    which is a tropical cone in $\mathcal{D}^{n+1}$. It is immediate that $\mathcal{P} = \mathcal{C} \cap \mathcal{D}^n$
\end{proof}

This lemma states that a tropical polyhedra can be represented by a tropical cone, for a tropical polyhedra in $\mathcal{D}^n$, we can represent it with two matrices $A,B \in \mathcal{D}^{q\times(n+1)}$ ($\mathcal{M}$-form), or its generating set $\{v_i\}_i \subseteq \mathcal{D}^{n+1}$. Some calculation concerning polyhedras can thus be passed to that concerning tropical cones. 
In the same way, we can define a more generalized form of DBMs, which we call an affine DBM, that are sets of the form $\mathscr{R}(M)\cap \mathcal{D}^n$ with $M \in \mathcal{R}^{(n+1)\times (n+1)}$. As tropical cones are union of finitely many closed DBMs, it is immediate that tropical polyhedras are union of finitely many closed affine DBMs.

\section{Problem Formulation}\label{sec: problem formulation}
In the rest of this paper,  the dioid $\mathcal{D}$ is supposed to be $\mathbb{R}_{\max}$. For clarity, we denote the natural order on the semimodule $\mathcal{D}^n$ or $\mathcal{D}^m$ by $\preceq$, and the order on $\mathcal{D}$ by $\le$.

We consider a max-plus linear system defined by
\begin{equation}\label{deterministic system}
    x_k = A \otimes x_{k-1}\oplus B \otimes u_k,
\end{equation}
where 
$x_{k}\in \mathcal{D}^{n}$ is the state at time $k$, $u_{k}\in\mathcal{U}\subseteq\mathcal{D}^m$ is the control input at time $k$ with $\mathcal{U}$ be the set of feasible control values, 
$A \in \mathcal{D}^{n\times n}$ and $B \in \mathcal{D}^{n\times m}$ are two matrices with coefficients in $\mathcal{D}$. 

Let  $E\subseteq \mathcal{D}^n$ be a region. 
The \textbf{one-step backward reachable set} of $E$ is defined by
\[\Upsilon_\mathcal{U}^{A,B}(E) = \{x\in \mathcal{D}^n\mid \exists u \in \mathcal{U}, A\otimes x \oplus B \otimes u \in E\}\]

In this paper, we are interested in the calculation of the reachable set of a special type of set $E$. 
Denote $\mathscr{Z}_n\subseteq 2^{\mathcal{D}^n}$ as subset of $\mathcal{D}^n$ defined recursively in the following way:
\[
S = \emptyset \mid \mathscr{H} \mid S^c \mid S_1\cup S_2 \mid S_1\cap S_2,
\]
where $\mathscr{H}$ denote an affine half-space of $\mathscr{H}$.
The our problem is formulated as follows. 

\begin{myprob}\label{main problem}
Let $\mathscr{S}$ denote the system described in \eqref{deterministic system}  with the control region $\mathcal{U}\in \mathscr{Z}_m$. 
Given $S \in \mathscr{Z}_n$ as the target set, 
determine the region 
$K = \Upsilon^{A,B}_\mathcal{U}(S)$ (resp. $K^{(N)} = (\Upsilon^{A,B}_\mathcal{U})^N(S)$) as the backward reachable set in one step (resp. $N$ step) of $S$.
\end{myprob}

\begin{remark}It is noteworthy that $\mathcal{S}$ can alternatively be defined as the union of finitely many affined DBMs (not necessarily closed), and every affined DBM can be recursively expressed in the aforementioned form. Nevertheless, the recursive definition is retained herein to ensure clarity in subsequent discussions.
\end{remark}

Note that our problem considered can be solved by the method of DBMs as mentioned in \cite{adzkiya2014forward}. Take, for example, the calculation of the set $\Upsilon_\mathcal{U}^{A,B}(S)$, the major steps are the following:
\begin{enumerate}
    \item Reform $S$ into a union of finitely many affined DBMs: $S = \bigcup_i X_i$
    \item Reform the system into an autonomous system $x_k = F\otimes y_{k-1}$, where $F = [A,B]$, $y_{k-1} = [x_{k-1}^T,u_k^T]^T$
    \item Compute the PWA system for $F$
    \item Compute the inverse image of each affined DBM by $F$ with the help of the PWA system, denoted as $F^{-1}(X_i)$, which is a finite union of affined DBMs in $\mathcal{D}^{n+m}$
    \item Compute the intersection of $F^{-1}(X_i)\cap \mathcal{D}^n\times \mathcal{U}$
    \item Compute the projection of each $F^{-1}(X_i)\cap \mathcal{D}^n\times \mathcal{U}$ on $\mathcal{D}^n$, denoted as $Y_i$
    \item Compute the union $\bigcup_iY_i$
\end{enumerate}
However, the time complexity and space complexity of this method is extremely large  due to the exponential combinations. 
Therefore, in this work, we provide a new method with lower complexity but with some approximation to solve the problem.

\section{Basic ideas}\label{sec: Basic ideas}
In this section, we present the basic idea for approximating the backward reach set described in Section~\ref{sec: problem formulation}. As a representative example, consider the computation of the set $\Upsilon_\mathcal{U}^{A,B}(E)$. The key steps of our method are as follows:
(i) We first transform the system $x_k = F \otimes y_{k-1}$ into an autonomous form, following the approach used in the DBM method.
(ii) We compute the pre-image of the target set $E$ under $F$, denoted as $F^{-1}(E)$. We then intersect this set with $\mathcal{D}^n \times \mathcal{U}$. The resulting set contains all pairs $(x, u)$ such that, starting from state $x$ and applying control $u \in \mathcal{U}$, the system transitions into the target region $E$ at the next time step.
(iii) We then compute an approximation of the projection $\mathbb{P}_n^{n+m}(F^{-1}(E))$. This approach allows us to avoid computing the projection $\mathbb{P}_n^{n+m}(F^{-1}(E))$ directly, which, to the best of our knowledge, can only be done via piecewise-affine (PWA) system techniques and incurs significant computational complexity. To approximate the resulting reachable set, we choose tropical polyhedra due to their compact computer representation and favorable mathematical properties. To approximate $\mathbb{P}_n^{n+m}(F^{-1}(E))$, we introduce the notion of closure approximation: we show later in the paper that the closure of this set, $\overline{\mathbb{P}_n^{n+m}(F^{-1}(E))}$, can be expressed as a finite union of tropical polyhedra. A key technical observation is that the $\mathcal{V}$-form or $\mathcal{M}$-form of $\overline{\mathbb{P}_n^{n+m}(F^{-1}(E))}$ cannot be computed directly. However, we resolve this issue by proving that
\(
\overline{\mathbb{P}_n^{n+m}(F^{-1}(E))} = \mathbb{P}_n^{n+m}(\overline{F^{-1}(E)}),
\)
where the right-hand side is much easier to compute.

\section{Approximation of Sets by Tropical Polyhedras}\label{sec: approximation method}
A principal challenge of our idea is how to give a proper approximation for a set in $\mathscr{Z}_n$ by polyhedras. In this section, we will describe the technical tool needed for such an approximation, justify its reasonability and provides an algorithm in the end for calculating such an approximation. 

We first specialize a kind of subset of $\mathcal{D}^n$: $\mathscr{X}_n \subset 2^{\mathcal{D}^n}$ formed by subset of $\mathcal{D}^n$ that is a union of finitely many tropical polyhedras. Since all sets in $\mathscr{X}_n$ is a closed set, we have $\mathscr{X}_n\subsetneq \mathscr{Z}_n$. Now given a set $S \in \mathscr{Z}_n$, our goal is to design a relatively faster algorithm that approximate $S$ with sets of $\mathscr{X}_n$ but with some errors that are negligible.

We define an auxiliary set $\mathscr{Z}_n' \subseteq 2^{\mathcal{D}^n}$ whose element are subset of $\mathcal{D}^n$ defined in the following way:
\[S = \emptyset \mid \mathscr{H} \mid S^c \mid S_1\cup S_2 \mid S_1\cap S_2\]
where $\mathscr{H}$ denote a half-space of $\mathcal{D}^n$. Similarly, define $\mathscr{X}_n'\subseteq 2^{\mathcal{D}^n}$ as the set of all subset of $\mathcal{D}^n$ that are union of finitely many tropical cones in $\mathcal{D}^n$.

By lemma \ref{instection with Dn} , we can prove that every set $X \in \mathscr{Z}_n$ (resp. $X \in \mathscr{X}_n$) can be expressed as $X = Y \cap \mathcal{D}^n$ for some $Y \in \mathscr{Z}_n'$ (resp. $Y \in \mathscr{X}_n'$). Thus intuitively, if we can find an approximation of $X \in \mathscr{Z}_n'$ by a set $X'\in\mathscr{X}_n'$, then the set $X\cap \mathcal{D}^n \in \mathscr{Z}_n$ can be approximated by $X' \cap \mathcal{D}^n \in \mathscr{X}_n$. The problem is thus transformed to approximating sets in $\mathscr{Z}_n'$ by sets in $\mathscr{X}_n'$. Notice that every set in $\mathscr{Z}_n'$ is a finite union of DBMs, while every set in $\mathscr{X}_n'$ is a finite union of closed DBMs. As each DBM can be written as intersection of set of the form $\{x \in \mathcal{D}^n\mid x_i \sim_{i,j} x_j + \alpha_{i,j}\}$, where $\sim_{i,j}$ is $<$ or $\le$. A traditional way of approximating a DBM by a closed DBM is to change all open sets of the form $\{x \in \mathcal{D}^n\mid x_i < x_j + \alpha_{i,j}\}$ into set of the form $\{x \in \mathcal{D}^n\mid x_i \le x_j + \alpha_{i,j}\}$, and the error of such an approximation is just the set $\{x\in \mathcal{D}^n \mid x_i = x_j + \alpha_{i,j}\}$, which is empty if $\alpha_{i,j} = +\infty$ and is a negligible hyperplan otherwise.

The previous discussion motivate us to defined the following:
\begin{mydef}
    Let $M \in \mathcal{R}^{n\times n}$ be a DBM, we define the closed form of $M$, denoted as $\overline{M}$, by replacing all the $<$ appearing in the expression of $M$ by $\le$.
\end{mydef}
For a simple DBM $\mathscr{R}(M)$, it can be approximated by $\mathscr{R}(\overline{M})$. Further more, the following proposition tell us that our approximation method is also ``topologically good":
\begin{mypro}\label{closure of DBM}
    If $\mathscr{R}(M)$ is non empty, then we have $\mathscr{R}(\overline{M}) = \overline{\mathscr{R}(M)}$.
\end{mypro}
\begin{proof}
    First note that for all $x_1,x_2 \in \mathscr{R}(M)$, we have $x_1\oplus x_2 \in \mathscr{R}(M)$ and $\mathrm{supp}(x_1 \oplus x_2) = \mathrm{supp}(x_1)\cup \mathrm{supp}(x_2)$. We can thus find an $x \in \mathscr{R}(M)$ that is of full support: for all index $i \le n$, $x_i = \varepsilon \Rightarrow \forall y \in \mathscr{R}(M), y_i = \varepsilon$. To simplify without loss of generation, we suppose that $\mathrm{supp}(x) = [n]$, in this case the term $(-\infty,\le)$ will not appear in the coefficient of $M$. We first prove that $\overline{\mathscr{R}(M)}\cap\mathbb{R}^n = \mathscr{R}(\overline{M})\cap \mathbb{R}^n$. The $\subseteq$ is immediate. Let $y \in \mathscr{R}(\overline{M})\cap \mathbb{R}^n$, and conxsider the sequence $\left\{z_k = \frac{1}{k}.x + \left(1-\frac{1}{k}\right).y\right\}_{k\in \mathbb{Z}_{>0}} $ where ``$.$" denote the usual multiplication by positive scalar. The sequence is contained in $\mathscr{R}(M)\cap \mathbb{R}^n$, this is because if $x_i < x_j + \alpha_{i,j}$ and $y_i \le y_j + \alpha_{i,j}$ for some $\alpha_{i,j} \in \mathbb{R}\cup\{+\infty\}$, then  $(\lambda.x_i + (1-\lambda).y_i) < (\lambda.x_j + (1-\lambda).y_j)+\alpha_{i,j}$. As $z_k$ converges to $y$, we conclude that $\overline{\mathscr{R}(M)}\cap\mathbb{R}^n \supseteq \mathscr{R}(\overline{M})\cap \mathbb{R}^n$. Now we prove that $\overline{\mathscr{R}(M)} = \mathscr{R}(\overline{M})$. It suffice to prove $\overline{\mathscr{R}(M)} \supseteq \mathscr{R}(\overline{M})$. Let $x$ be in $\mathscr{R}(\overline{M})$ and we wish to show that $x \in \overline{\mathscr{R}(M)}$. Take $y \in \mathscr{R}(M)\cap \mathbb{R}^n$, then the sequence $\{z_k = x \oplus (-k)\cdot y\}_{k \in \mathbb{Z}_{>0}}$ converge to $x$ (the ``$\cdot$" denote multiplication in max-plus sense), and as  $(-k)\cdot y \in \mathbb{R}^n\cap \mathscr{R}(M) \subseteq\mathbb{R}^n\cap \mathscr{R}(\overline{M})$, we have $z_k \in \mathscr{R}(\overline{M})\cap \mathbb{R}^n = \overline{\mathscr{R}(M)}\cap\mathbb{R}^n$. $x$ is the limit of $z_k$, thus $x \in \overline{\mathscr{R}(M)}$.
 \end{proof}
 We can now derive our first method of approximating sets in $\mathscr{Z}_n'$ by set in $\mathscr{X}_n'$: take a set $X \in \mathscr{Z}_n'$, denote $X$ as $X = \bigcap_i \mathscr{R}(M_i)$ with all $\mathscr{R}(M_i)$ non empty, then an approximation for $X$ is given by $X' = \bigcup_i \mathscr{R}(\overline{M_i})$, and what's more, $X'$ is exactly the topological closure of $X$. 

Although we succeed in approximating $X$ by a union of tropical cones (since every DBM is a tropical cone), this method can be very complicated since the number of DBMs needed to represent a tropical cone can be exponentially large. In the following we will represent a less complicated approximation method, which is in fact equivalent to this first method described here. We first present the following useful lemma which is essential of our approximation theory.
 \begin{mylem}\label{closure of DBM2}
     Let $X$ be a union of finitely many DBMs, and $Y$ be a union of finitely many 
     open DBMs (all DBM are of region non empty). Then we have 
     \[\overline{X\cap Y} = \overline{\overline{X}\cap Y}\]
 \end{mylem}
 \begin{proof}
     Denote $X = \bigcup_{i=1}^r\mathscr{R}(M_i)$ and $Y = \bigcup_{j=1}^s\mathscr{R}(N_j)$, with $N_i$ open DBMs. Then we have $X \cap Y = \bigcup_{i,j}$, $\overline{X} = \bigcup_i\overline{\mathscr{R}(M_i)}$ and $\overline{X}\cap Y = \bigcup_{i,j}\overline{\mathscr{R}(M_i)}\cap \mathscr{R}(N_j)$. It suffice to prove that for all $i,j$, we have $\overline{\mathscr{R}(M_i) \cap \mathscr{R}(N_j)} = \overline{\overline{\mathscr{R}(M_i)}\cap \mathscr{R}(N_j)}$. As $N_j$ is a open DBM, the region associated $\mathscr{R}(N_j)$ is also open. We have thus $\mathscr{R}(M_i) \cap \mathscr{R}(N_j)$ is non empty iff $\overline{\mathscr{R}(M_i)}\cap \mathscr{R}(N_j)$ is non empty. The case both sets are empty is trivial. Suppose that they are both non empty, then we have  $\overline{\mathscr{R}(M_i) \cap \mathscr{R}(N_j)} = \mathscr{R}(\overline{M_i\oplus N_j}) = \mathscr{R}(\overline{M_i}\oplus \overline{N_j})$, and $\overline{\overline{\mathscr{R}(M_i)}\cap \mathscr{R}(N_j)} = \overline{\mathscr{R}(\overline{M_i})\cap \mathscr{R}(N_j)} = \mathscr{R}(\overline{\overline{M_i}\oplus N_j}) = \mathscr{R}(\overline{M_i}\oplus \overline{N_j})$.
 \end{proof}

Note that all set in $\mathscr{Z}_n'$ can be rewrite as a finite union of sets of the form $E = \mathscr{H}_1 \cap \dots \cap \mathscr{H}_{k_1} \cap \mathscr{H}_{k_1+1}^c\dots \mathscr{H}_{k_1+k_2}^c$, where $\mathscr{H}_i$ are tropical half-spaces. It suffice for us to give an approximation for each one of these sets. Our principle idea is to do the calculation step by step: first we calculate a $\mathcal{V}$-form of $\overline{\mathscr{H}_1 \cap \dots \cap \mathscr{H}_{k_1} \cap \mathscr{H}_{k_1+1}^c}$, then we calculate that of $\overline{\overline{\mathscr{H}_1 \cap \dots \cap \mathscr{H}_{k_1} \cap \mathscr{H}_{k_1+1}^c}\cap \mathscr{H}_{k_1+1}^c}$, we iterate until we meet $\mathscr{H}_{k_1+k_2}^c$. As $\mathscr{H}_{k_1+i}^c$ is a finite union of  open DBMs, we can simplify the set thanks to lemma \ref{closure of DBM2}, and finally we will get the closure of $E$ as the approximation of $E$ with this method, thus the novel method is equivalent to our first method. The following theorem tells us how calculate the $\mathcal{V}$-form of intermediate sets such as $\overline{\overline{\mathscr{H}_1 \cap \dots \cap \mathscr{H}_{k_1} \cap \mathscr{H}_{k_1+1}^c}\cap \mathscr{H}_{k_1+1}^c}$.

\begin{mythm}\label{H^q finitely generated}
Let $\mathscr{C} \subseteq \mathcal{D}^n$ be a tropical cone with $V_0$ its set of generating vectors. Let $\mathscr{H} = \{x \in \mathcal{D}^n\mid (a|x) \le (b|x)\}$ be the half-space in $\mathcal{D}^n$ where $a \in \mathcal{D}^n\setminus \{\varepsilon_n\}$. We suppose in addition that $\mathscr{C} \cap \mathscr{H}^c \cap \mathbb{R}^n \neq \emptyset$. Then the topological closure of the tropical cone $\mathscr{C}\cap \mathscr{H}^c$ is generated by the following set
\begin{align*} 
 V_\infty = & \{v \in V_0 \mid (b|v) < (a|v)\}\cup \bigg\{ (b|w)\cdot v \oplus (a|v)\cdot w \ \bigg| \\ 
& \quad\quad v, w \in V_0, \ (b|v) < (a|v),\ (a|w) \leq (b|w) \bigg\}
\end{align*}
i.e. $\mathrm{Span}(V_\infty) = \overline{\mathscr{C}\cap \mathscr{H}^c}$.

\end{mythm}
\begin{proof}
Before proving this theorem, we first provided a fact that is useful in the following:
Let $\mathscr{F} = \{x \in \mathcal{D}^n\mid (a|x) \le (b|x)\}$ be a half-space, a generating set for the tropical cone $\mathscr{C}\cap \mathscr{F}$ is given by 
    \begin{align*}
        V& = \{v \in V_0\mid v \in \mathscr{F}\}\cup \\& \{(a|w)v\oplus (b|v)w\mid v,w \in V_0,\ v \in \mathscr{F},\ w \notin \mathscr{F}\}
    \end{align*}
This is a direct application of \cite[theorem 14]{allamigeon2013computing}.

Now we prove the theorem. Firstly, one verifies quickly that 
    \[\mathscr{H}^c\cap \mathbb{R}^n = \{x \in \mathcal{D}^n\mid (a|x) > (b|x)\}\cap \mathbb{R}^n = \bigcup_{k\in \mathbb{Z}_{>0}}\mathscr{F}_k\cap \mathbb{R}^n\]
     where $\mathscr{F}_k$ are half-spaces defined by 
    \begin{align*}
        \mathscr{F}_k = \{x \in \mathcal{D}^n\mid(a|x)\ge (b|x)+1/k\} \\= \{x\in \mathcal{D}^n\mid (a|x)\ge (b^{(k)}|x)\}
    \end{align*}
    and $b^{(k)}$ is defined as $b^{(k)} = (b_1+1/k,\dots,b_n+1/k)$. This is because for $x \in \mathbb{R}^n$, we have $(a|x) \in \mathbb{R}$, thus $(a|x) > (b|x)$ is equivalent to $\exists k \in \mathbb{Z}_{>0},\ (a|x) \ge (b|x)+1/k$. Denote $F=\left(\bigcup_{k\in \mathbb{Z}_{>0}}\mathscr{F}_k\right)\cap \mathscr{C}$, then we have $F \cap\mathbb{R}^n = \mathscr{C}\cap \mathscr{H}^c\cap \mathbb{R}^n$. By taking the closure in $\mathbb{R}^n$, we have $\overline{F \cap \mathbb{R}^n} = \overline{\mathscr{C}\cap \mathscr{H}^c \cap \mathbb{R}^n}$. Remark that $\mathbb{R}^n$ is an open set in the metrizable space $\mathcal{D}^n$, let $\overline{F}$ denote the closure of $F$ in $\mathcal{D}^n$, we then have $\overline{F}\cap \mathbb{R}^n = \overline{F\cap \mathscr{C}}$ and similarly, $\overline{\mathscr{C}\cap \mathscr{H}^c\cap \mathbb{R}^n} = \overline{\mathscr{C}\cap \mathscr{H}^c}\cap \mathbb{R}^n$.
    
    We have 
    \[\overline{F} \cap \mathbb{R}^n= \overline{\mathscr{C}\cap \mathscr{H}^c}\cap \mathbb{R}^n \neq \emptyset\]
    
    We can also derive that $\overline{F} = \overline{\mathscr{C}\cap \mathscr{H}^c}$. This is due to the fact that if $P$, $Q$ are two closed semimodules in $\mathcal{D}^n$ such that $P\cap \mathbb{R}^n = Q \cap \mathbb{R}^n \neq \emptyset$, then $P = Q$. To prove this, let $x \in P\cap \mathbb{R}^n = Q \cap \mathbb{R}^n$ and $y \in P$, then the subsequence $\{(-k).x\oplus y \}_{k\in \mathbb{Z}_{>0}} \subseteq \mathbb{R}^n\cap P = \mathbb{R}^n\cap Q$ converge to $y$, which means that $y \in Q$. It suffice for us to find the generating set for $F=\overline{\left(\bigcup_{k\in \mathbb{Z}_{>0}}\mathscr{F}_k\right)\cap \mathscr{C}}$.
    
    We remark that for all $k\in \mathbb{Z}_{>0}$ we have $\mathscr{F}_k\subset \mathscr{F}_{k+1}$. Define $V_k$ as
\begin{align*} 
 V_k = & \{v \in V_0 \mid (b|v) < (a|v)\} \\ 
& \quad \cup \bigg\{ \left((b|w) + \frac{1}{k}\right)\cdot v \oplus (a|v)\cdot w \ \bigg| \\ 
& \quad\quad v, w \in V_0, \ (b|v) < (a|v),\ (a|w) \leq (b|w) \bigg\} \end{align*}
For $v\in V_0$ such that $(b|v) < (a|v)$, there exists $N_v > 0$ such that for all $k \ge N_v$, we have $v \in \mathscr{F}_k$, and for all $k \in \mathbb{Z}_{>0}$, we have $(a|w) \le (b|w) \Rightarrow w \notin \mathscr{F}_k$. Now let $N = \max_{v\in V_0,\ (b|v)<(a|v)}(N_v)$, then $N< +\infty $ since set $V_0$ is finite. From the fact presented in the beginning of the proof, we know that for $k>N$, the finite set $V_k$ is a generator of $\mathscr{C}\cap\mathscr{F}_k$. Now consider the tropical cone $\mathrm{Span}(V_\infty)$. First we prove that $\mathrm{Span}(V_\infty) \supseteq F$. Let $x \in \left(\bigcup_{k \in \mathbb{Z}_{>0}}\mathscr{F}_k\right)\cap \mathscr{C} $, then there exist $k\in \mathbb{Z}_{>0}$ such that $x \in \mathscr{F}_k\cap \mathscr{C}$. We can suppose that $k$ is large enough that $\mathscr{C}\cap\mathscr{F}_k$ is generated by $V_k$. $x$ can thus be written in the following form
\[x = \bigoplus_{v}\lambda_v.v \oplus \bigoplus_{v,w}\lambda_{v,w}.\left[\left((b|w)+\frac{1}{k}\right).v \oplus (a|v).w\right]\]
where $\lambda_v$ and $\lambda_{v,w}$ are elements in $\mathcal{D}$ indexed by the sets $\{v \in V_0\mid (b|v)<(a|v)\}$ and $\{w \in V_0\mid (b|v)\ge(a|v)\}$. We then proceed to rewrite $x$ as a combination of elements in $V_\infty$:
\begin{equation}\label{x in spanV}
    x = \bigoplus_{v}\lambda_v'.v \oplus\bigoplus_{v,w}\lambda_{v,w}.\left[(b|w).v \oplus (a|v).w\right]
\end{equation}
where $\lambda_v'$ is defined as \(\lambda_v' =\lambda_v \oplus \bigoplus_w(\lambda_{v,w}+(b|w)+1/k)\), which proves that $x \in \mathrm{Span}(V_\infty)$. It remains to prove that $\mathrm{Span}(V_\infty) \subset \overline{\mathscr{C}\cap\mathscr{H}^c}$. Take an element $x\in \mathrm{Span}(V_\infty)$ described as in formula \ref{x in spanV}, and define $x_k$ as
\[x_k = \bigoplus_{v}\lambda_v'.v \oplus \bigoplus_{v,w}\lambda_{v,w}.\left[\left((b|w)+\frac{1}{k}\right).v \oplus (a|v).w\right]\]
then for $k > N$, $x_k \in \mathscr{F}_k \subseteq \mathscr{H}^c$. $\{x_k\}_{k\ge N}$ is thus a sequence in $\mathscr{C}\cap \mathscr{H}^c$ By continuity of the max-plus operations $\max$ and $+$, we have $x_k \xrightarrow{k \to \infty} x$, which proves that $x \in \overline{\mathscr{C}\cap\mathscr{H}^c}$.
    \end{proof}
\begin{remark}
We remark that a critical condition of theorem \ref{H^q finitely generated} is that the set $\mathscr{C}\cap \mathscr{H}^c \cap \mathbb{R}^n$ is supposed to be non empty. We shall be able to determine this kind of situation in practical use. Remark that the fact $\mathrm{Span}(V_\infty)\cap \mathbb{R}^n = \overline{\mathscr{C}\cap \mathscr{H}^c}\cap \mathbb{R}^n$ still holds even if we remove this non-empty constraint. Thus we can check the emptiness by regarding the set $V_\infty$, if $\bigoplus_{v\in V_\infty}v$ is not of full support, meaning that $\mathrm{Span}(V_\infty)\cap \mathbb{R}^n = \overline{\mathscr{C}\cap \mathscr{H}^c} \cap \mathbb{R}^n = \emptyset$, the later is equivalent to $\mathscr{C}\cap \mathscr{H}^c \cap \mathbb{R}^n = \emptyset$. If unfortunately, we encounter cases where $\mathscr{C}\cap \mathscr{H}^c \cap \mathbb{R}^n = \emptyset$, we can restrict $\mathscr{C}$ and $\mathscr{H}^c$ to smaller support such that the hypotheses of theorem \ref{H^q finitely generated} becomes valid. We describe how this works later in an algorithm.        
\end{remark}
\begin{remark}
For $\mathscr{H} = \langle a,b \rangle$ such that $a = \varepsilon_n$, then $\mathscr{H} = \mathcal{D}^n$ and $\mathscr{H}^c = \emptyset$, it is thus a trivial case.
\end{remark}
We now provided an example showing the necessity of considering the set $\overline{\mathscr{C}\cap\mathscr{H}^c}$:
\begin{myexm}
    Consider the tropical half-space $\mathscr{H}$ given by $\mathscr{H} = \{(x_1,x_2) \in \mathbb{R}_{\max}^2\mid \max(x_1+1,x_2)\le \max(x_1,x_2)\}$. It is immediate to verify that $\mathscr{H}$ is in fact the set $\{(x_1,x_2)\in \mathbb{R}_{\max}\mid x_1+1\le x_2\}$. Note that the half-space $\mathscr{H}' =\{(x_1,x_2) \in \mathbb{R}_{\max}^2\mid \max(x_1+1,x_2)\ge \max(x_1,x_2)\}$ which keeps the same definition of $\mathscr{H}$ except that we have switched the $\le$ in the formula into $\ge$. One can find out quickly that the set $\mathscr{H}'$ is exactly the whole space $\mathbb{R}_{\max}^2$, which is of no interest and has nothing to do with the set $\mathscr{H}^c$. Now let's follow the method in the proof of lemma \ref{H^q finitely generated} and search a generating set for $\overline{\mathscr{H}^c}$. Now the set $V_0 = \{(0,\varepsilon),(\varepsilon,0)\}$, let $a = (1,0)$, $b=(0,0)$, $v = (0,\varepsilon)$, $w = (\varepsilon, 0)$. One verifies quickly that \(V_\infty = \{(0,\varepsilon)\}\cup \{(b|w).v\oplus (a|v).w\} = \{(0,\varepsilon),(0,1)\}\), and $\mathrm{Span}(V_\infty) = \{(x_1,x_2)\in \mathbb{R}_{\max}^2\mid x_1 +1 \ge x_2\}$, which is the closure of $\mathscr{H}^c = \{x\in \mathbb{R}^2_{\max}\mid x_1+1 > x_2\}$.
\end{myexm}

Now for a set $E$ of the form $E = \mathscr{H}_1 \cap \dots \cap \mathscr{H}_{k_1} \cap \mathscr{H}_{k_1+1}^c\dots \mathscr{H}_{k_1+k_2}^c$ (where the $\mathscr{H}_i$s denote tropical half-spaces), we proceed the following method to compute the $\mathcal{V}$-form of a tropical cone that approximates the set $E$:
\begin{enumerate}
    \item We first compute a generating set $V$ of the tropical cone $\mathscr{C}_0 =  \mathscr{H}_1 \cap \dots \cap \mathscr{H}_{k_1}$ with the method described in \cite{allamigeon2013computing}. We let $J$ be the support of $\mathscr{C}_0$
    \item For $i$ ranging from $1$ to $k_2$, compute the set $V_\infty$ for the tropical cone $\mathscr{C}_i$ and the half-space $\mathscr{H}_{k_1 + i}$. If $v_\infty = \bigoplus_{v \in V_\infty}v$ is not of full support, i.e. $\mathrm{supp(v_\infty)} \subsetneq J$, we redo the calculation by restricting $\mathscr{C}$ and $\mathscr{H}^c$ on the support of $v_\infty$ until we finally get some $v_\infty$ of full support. We update $J$ by the support of $\mathrm{Span}(V_\infty) = \overline{\mathscr{C}_i\cap \mathscr{H}^c_{k_1+i}}$ and $V$ by $V_\infty$.
    \item Return the final $V$ which is a generating set for the restriction of $\mathscr{C}_{k_2}$ on the support $J$, the generating set for $\mathscr{C}_{k_2}$ is obtained by padding all vectors in $V$ into vectors of $\mathcal{D}^n$ with $\varepsilon$.
\end{enumerate}
Remark that the ``while" process in step 2 of the algorithm will terminate since the support $J$ is strictly decreasing. Further more, as we always have $\mathrm{Span}(V_\infty) \supseteq \overline{\mathscr{C}_i\cap \mathscr{H}^c_{k_1+i}}$, the support of $\mathscr{H}^c_{k_1+i}$ is always contained in the variable $J$ during the iteration, which will finally be $J$ in a finite number of steps. We have been cutting off the dimension of the system until that the condition required by theorem \ref{H^q finitely generated} is verified.

As $\mathscr{H}^c_{k_1+i}$ is a union of open DBMs, by lemma \ref{closure of DBM2}, we can prove by induction that $\overline{E} = \mathrm{Span}(V)$.

\section{The reachability analysis for MPLS with approximation}\label{sec: reachability}
In this section we focus on the calculation of reachable set $E$ as described in the section \ref{sec: Basic ideas}. We first show how to calculate the closure of $\Upsilon_{\mathcal{U}}^{A,B}(E)$ as its approximation. For the set $\left(\Upsilon_{\mathcal{U}}^{A,B}\right)^N(E)$, we remark that it is just the reachable set in one step of the system 
\[x_k = A' \otimes x_{k-1}\oplus B' \otimes u_k'\]
where $A' = A^{\otimes N}$ and $B' = [B|A^{\otimes N}\otimes B| A^{\otimes (N-1)} \otimes B|\dots|A\otimes B|B]$. 

Remark that $S$ can be rewritten into the union of finitely many set of the form $\mathscr{H}_1\cap \dots \cap \mathscr{H}_{k_1} \cap \mathscr{H}_{k_1+1}^c \dots \cap \mathscr{H}_{k_1+k_2}^c$, where each $\mathscr{H}_i$ is a tropical affine half-space. The latter can be rewritten into the form $E = \{x \in \mathcal{D}^n\mid C_1\otimes (0,x)  \preceq D_1 \otimes (0,x),\ C_2\otimes (0,x)  \prec D_2 \otimes (0,x) \oplus f_2\}$. The set $\Upsilon^{A,B}_\mathcal{U}(E)$ can be written as 
\[\Upsilon^{A,B}_\mathcal{U}(E) = \left\{x \in \mathcal{D}^n\left| 
\begin{aligned}
    &\exists u \in \mathcal{U},\ z = A\otimes x \oplus B \otimes u\\
    &C_1\otimes (0,z) \preceq D_1 \otimes (0,z)\\
    &C_2 \otimes (0,z) \prec D_2 \otimes (0,z)
\end{aligned}\right\}\right.\]
There exists matrices $F,G$ such that $\mathcal{U} = \{u \in \mathcal{D}^m\mid F \otimes (0,u) \preceq G \otimes (0,u)\}$. By reformulating the expressions, we can find matrices $K_1,K_2$ and $L_1,L_2$ such that 
$$\Upsilon_\mathcal{U}^{A,B}(E) = 
\left\{ x \in \mathcal{D}^n \left|\begin{aligned} &\exists u \in \mathcal{D}^n \\
& K_1 \otimes (0,x,u) \preceq L_1\otimes (0,x,u)\\
& K_2 \otimes (0,x,u) \prec L_2 \otimes (0,x,u)
\end{aligned}\right\}\right.$$
explicitly, $K_1$, $L_1$ are given by
\[K_1 = \left[
\begin{array}{c|c|c}
    C_1^{(1)} & C_1^{(\ge 2)} \otimes A & C_1^{(\ge 2)}\otimes B \\ \hline
    F^{(1)} & \mathcal{E}  & F^{(\ge 2)}
\end{array}
\right]\]
\[ L_1 = \left[
\begin{array}{c|c|c}
    D_1^{(1)} & D_1^{(\ge 2)} \otimes A & D^{(\ge 2)} \otimes B \\ \hline
    G^{(1)} & \mathcal{E} & F^{(\ge 2)}
\end{array}
\right]
\]
where the matrix $\mathcal{E}$ has all its coefficients equal to $\varepsilon$. The matrix $K_2$, $L_2$ are given by $K_2 = [C_2^{(1)},C_2^{(\ge 2)}\otimes A, C_2^{(\ge 2)} \otimes B]$, $K_2 = [D_2^{(1)},D_2^{(\ge 2)}\otimes A, D_2^{(\ge 2)} \otimes B]$.

From this point of view we can find out that $\Upsilon_\mathcal{U}^{A,B}(E)$ is in fact the projection on the first $n$-coordinates of the following set
\[\mathcal{P} = \left\{z \in \mathcal{D}^{n+m}\left| 
\begin{aligned}
    & K_1\otimes (0,z) \preceq L_1 \otimes (0,z) \\
    & K_2 \otimes (0,z) \prec L_2 \otimes (0,z)
\end{aligned}
\right\}\right.\]

 Let $\mathbb{P}_n^{n+m}$ be the projector of $\mathcal{D}^{n+m}$ on its first $n$-coordinates, then $\Upsilon_{\mathcal{U}}^{A,B}(E) = \mathbb{P}_n^{n+m}(\mathcal{P})$. $\mathcal{P}$ can be rewritten as $F \cap \mathcal{D}^{n+m}$ with $F\subseteq \mathcal{D}^{n+m+1}$ of the form $F = \mathscr{H}_1\cap \dots \cap \mathscr{H}_{k_1} \cap \mathscr{H}_{k_1+1}^c\dots \mathscr{H}_{k_1+k_2}^c$. Then $\mathbb{P}_n^{n+m}(\mathcal{P}) = \mathbb{P}_n^{n+m}(F\cap \mathcal{D}^{n+m})= \mathbb{P}_{n+1}^{n+m+1}(F)\cap \mathcal{D}^n$. 
 
 We aim to find an approximation for the set $\mathbb{P}_n^{n+m}(\mathcal{P})$. 
 As described in section \ref{sec: Basic ideas}, we would like to calculate the closure of $\mathbb{P}_{n+1}^{n+m+1}(F)$ by the alternative method of calculating firstly the closure $\overline{F}$ of $F$ then project $\overline{F}$ onto $\mathcal{D}^{n+1}$. This is because $\mathbb{P}_{n+1}^{n+m+1}(\overline{F})$ is exactly the closure of $\mathbb{P}_{n+1}^{n+m+1}$. Indeed we always have $\mathbb{P}_{n+1}^{n+m+1}(\overline{F}) \supseteq \overline{\mathbb{P}_{n+1}^{n+m+1}(F)}$, because $\mathbb{P}_{n+1}^{n+m+1}(\overline{F})$ is a tropical cone (thus closed) containing $\mathbb{P}_{n+1}^{n+m+1}(F)$. And for $x\in \mathbb{P}_{n+1}^{n+m+1}(\overline{F})$, there exists $u \in \mathcal{D}^m$ and a sequence $(x_k,u_k)_k \in F$ such that $(x_k,u_k)\xrightarrow{k\to +\infty}(x,u)$. As $x_k \in \mathbb{P}_{n+1}^{n+m+1}(F)$, we have $x \in \overline{\mathbb{P}_{n}^{n+m}(F)}$. With this reasoning, we have $\overline{\mathbb{P}_n^{n+m}(\mathcal{P})} = \overline{\mathbb{P}_{n+1}^{n+m+1}(F)\cap \mathcal{D}^n} = \mathbb{P}_{n+1}^{n+m+1}(\overline{F})\cap \mathcal{D}^n = \mathbb{P}_n^{n+m}(\overline{\mathcal{P}})$. The following diagram illustrates this fact more clearly.
 \[
 \begin{tikzcd}
 F\arrow[r,"\mathrm{closure}"] \arrow[d,"\mathbb{P}^{n+m}_n"] & \overline{F} \arrow[d,"\mathbb{P}^{n+m}_n"] \\
 \mathbb{P}^{n+m}_n(F)\arrow[r,"\mathrm{closure}"] & \mathbb{P}^{n+m}_n(\overline{F}) 
\end{tikzcd}
\]
The $\mathcal{V}$-form of the set $\overline{\mathbb{P}^{n+m}_n(\mathcal{P})}$ can thus be obtained in the following steps: 1. Get the $\mathcal{V}$-form of $\overline{\mathcal{P}} = \overline{F}\cap \mathcal{D}^{m+n}$ through the method described in section \ref{sec: approximation method}. 2. For every generating vector of $\overline{\mathcal{P}}$, project it onto the first $n$ coordinates, the collection of these vectors is then a $\mathcal{V}$-form of $\overline{\mathbb{P}^{n+m}_n(\mathcal{P})}$ which is an approximation of $\overline{\mathbb{P}^{n+m}_n(\mathcal{P})}$. 

\begin{remark}
    A remarkable fact is that for the case $E = \mathscr{H}_1\cap \dots \cap \mathscr{H}_{k_1}$ (meaning that $k_2 = 0$), then we have $\Upsilon_\mathcal{U}^{A,B}(E) = \overline{\Upsilon_\mathcal{U}^{A,B}(E)}$, there is no approximation. The control strategy is stored in the $\mathcal{V}$-form of $\mathcal{P}$. Denote $(\{v_i\}_i,\{e_j\}_j)$ the $\mathcal{V}$-form of $\mathcal{P}$, then (as described in \cite{4099499}) given $x \in \Upsilon_\mathcal{U}^{A,B}(S)$, by the residuation theory we can find $(\lambda_i)_i,(\mu_j)_j$ such that $\bigoplus_j\mu_j = e$ and $x = \bigoplus_i\lambda_i\cdot v_i^{(\le n)} \oplus \bigoplus_j\mu_j\cdot e_j^{(\le n)}$. A possible control $u$ is then given by $u = \bigoplus_i\lambda_i\cdot v_i^{(\ge n+1)} \oplus \bigoplus_j\mu_j\cdot e_j^{(\ge n+1)}$. However, for the case $E$ is not a closed set, i.e. $k_2 \neq 0$, the control strategy $u$ obtained by the residuation theory can not garantee that $A\otimes x \oplus B \otimes u \in E$.
\end{remark}

\section{Improvement of the algorithm}\label{sec: Improvement}
 A point $x$ in the tropical cone $\mathscr{C}$ is said to extremal if $\forall a,b \in\mathscr{C}$, $x = a \oplus b \Rightarrow x = a$ or $x=b$. Every tropical cones is generated by its scaled extremal points, when calculating the $\mathcal{V}$-form of $\mathscr{C}$, we simply need to keep the extremal points in the representation, which can reduce the complexity to a large extent. \cite{allamigeon2013computing} have provided an efficient way of deciding whether a point $v \in \mathscr{C} = \langle C,D\rangle$ ($C,D \in \mathcal{D}^{p\times n}$) is a extremal point by considering the tangent cone at $v$. Here we give the definition of the tropical version of tangent cone:
 \begin{mydef}
     Given a subsemimodule $\mathscr{D}\subseteq \mathcal{D}^n$, and $v \in \mathscr{D}$ a vector, then the tangent cone $\mathscr{T}(v,\mathscr{D})$ is a subsemimodule in $\mathcal{D}^n$ containing $0$ and stable under the usual multiplication by positive scalars and is s.t. there exists a neighborhood $\mathcal{N}$ of $v$, $\mathcal{N} \cap \mathscr{D} = \mathcal{N} \cap \mathscr{T}(v,\mathscr{D})$.
 \end{mydef}
Remark that this definition does not depend on the choice of $\mathcal{N}$ and the tangent cone at a point is unique. For the tropical cone $\mathscr{C}$ and $v \in \mathscr{C}$ we have $\mathscr{T}(v,\mathscr{C}):= $
\[\bigcap_{
     r,\ (C_r| v) = (D_r | v) > \varepsilon
 }\left\{ x \!\in \!\mathcal{D}^n \left|\bigoplus_{i \in \mathrm{argmax}(C_r|v)}\!\!\!\!\!\!\!x_i\le\!\!\!\!  \bigoplus_{j \in \mathrm{argmax}(D_r|v)}\!\!\!\!\!\!x_j \right\}\right.\]
The set $\mathrm{argmax}(c | v)$ is defined as $\{1\le i \le n\mid c_i\otimes v_i = (c|v)\}$. Whether $v \in \mathscr{C}$ is an extremal point is totally decide by $\mathscr{T}(v,\mathscr{C})$. In fact, $v$ is extremal iff $0$ is extremal in $\mathscr{T}(v,\mathscr{C})$. Let $\mathcal{G}(v,\mathscr{C})$ be the hypergraph associated to $\mathscr{T}(v,\mathscr{C})$ as defined in \cite{allamigeon2013computing}, then $v \in \mathscr{C}$ is an extremal point iff $\mathcal{G}(v,\mathscr{C})$ has a maximal strongly connected componenet.

As showed in \cite{allamigeon2013computing}, we can suppose for simplicity that $v$ is of full support.

We wish to decide for $v \in \overline{\mathscr{C}\cap \mathscr{H}^c}$, whether it is an extremal point. We distinguish two cases: 1. The vector $v$ is contained in $\mathscr{H}^c$, i.e. $(b|v) < (a|v)$. 2. The vector $v$ is s.t. $(a|v) = (b|v)$. For the first case, $v$ is an extremal point in $\overline{\mathscr{C}\cap \mathscr{H}^c}$ iff $v$ is an extremal point in $\mathscr{C}$. In fact, let $\mathscr{T}(v,\mathscr{C})$ be the tangent cone of $\mathscr{C}$ at $v$, and define $\mathcal{N}$ as the open neighbourhood of $v$ formed by $x \in \mathcal{D}^n$ verifying the followings:
\begin{enumerate}
    \item $(C_r|x) < (D_r|x)$ for all $r \in [p]$ such that $(C_r|v) < (D_r|v)$.
    \item $\mathrm{argmax}(C_r|x)\subseteq \mathrm{argmax}(C_r|v)$ and $\mathrm{argmax}(D_r|x)\subseteq \mathrm{argmax}(D_r|v)$.
    \item $(b|x) < (a|x)$.
\end{enumerate}
then it can be verfied that for $x \in \mathcal{N}$, $x$ belongs to $\mathscr{C}\cap \mathscr{H}^c$ iff $x \in v + \mathscr{T}(v,\mathscr{C})$. Since $v + \mathscr{T}(v,\mathscr{C})$ is a closed set, by taking the closure we have $ \overline{\mathscr{C}\cap \mathscr{H}^c} \cap \mathcal{N} = (v + \mathscr{T}(v,\mathscr{C}))\cap \mathcal{N}$. This means that the tangent cone of $\overline{\mathscr{C}\cap \mathscr{H}^c}$ at $v$ is exactly $\mathscr{T}(v,\mathscr{C})$. That $v$ is extremal or not in $\overline{\mathscr{C}\cap \mathscr{H}^c}$ is totally the same as for $v$ in $\mathscr{C}$.

The difficulty lies in the second case. To the author's knowledge, we can only decide quickly only a part of vector that is extremal, for the rest we will have to use the brutal force to decide wether they are extremal points or not. Suppose now $v \in \overline{\mathscr{C}\cap \mathscr{H}^c}$ is one of the vector of $V_\infty$ provided by theorem \ref{H^q finitely generated}. We wish to determine whether $v$ is extremal in $\overline{\mathscr{C}\cap \mathscr{H}^c}$. The tangent cone of $\mathscr{C}\cap \mathscr{H}^c$ at $v$ have the following expression:
\[\mathscr{T}(v,\mathscr{C}\cap \mathscr{H}^c) = \mathscr{T}(v,\mathscr{C})\cap \mathscr{T}(v,\mathscr{H}^c)\]
and $ \mathscr{T}(v,\mathscr{H}^c) = \{x \in\mathcal{D}^n\mid \bigoplus_{i \in \mathrm{argmax}(b|v)}x_i < \bigoplus_{j \in \mathrm{argmax}(a|v)}x_j\}$. To prove this, first notice that the tangent cone of the intersection of two semimodules at a point $v$ is the intersection of the tangent cones of each semimodule at $v$, it suffice to calculate the tangent cone $\mathscr{T}(v,\mathscr{H}^c)$. Consider the open set $\mathcal{N}$ whose elements are $x \in \mathcal{D}^n$ such that $\mathrm{argmax}(b|cx) \subseteq \mathrm{argmax}(b|v)$ and $\mathrm{argmax}(a|x) \subseteq \mathrm{argmax}(a|v)$, it can be verified that $\mathscr{H}^c \cap \mathcal{N} = \mathscr{H}^c \cap \{x \in\mathcal{D}^n\mid \bigoplus_{i \in \mathrm{argmax}(b|v)}x_i < \bigoplus_{j \in \mathrm{argmax}(a|v)}x_j\}$. However it is not easy to decide the form of the tangent cone $\mathscr{T}\left(v,\overline{\mathscr{C}\cap \mathscr{H}^c}\right)=\overline{\mathscr{T}(v,{\mathscr{C}\cap \mathscr{H}^c})}$, to deal with this, we will enlarge the set a little, with the risk of loosing some authenticity. Here we describe our method: Note that the set $\mathscr{T}(v,\mathscr{H}^c) = \{x \in\mathcal{D}^n\mid \bigoplus_{i \in I}x_i < \bigoplus_{j \in J}x_j\}$ (where $I = \mathrm{argmax}(b|v)$, $J = \mathrm{argmax}(a|v)$) can be rewritten as $\mathscr{T}(v,\mathscr{H}^c) = \{x \in\mathcal{D}^n\mid \bigoplus_{i \in J}x_i < \bigoplus_{j \in J \setminus I}x_j\}$, whose closure can be written as $\overline{\mathscr{T}(v,\mathscr{H}^c)} = \{x \in\mathcal{D}^n\mid \bigoplus_{i \in I}x_i \le \bigoplus_{j \in J\setminus I}x_j\}$. As we have $\overline{\mathscr{T}(v,\mathscr{C}\cap \mathscr{H}^c)} \subseteq \mathscr{T}(v,\mathscr{C})\cap \overline{\mathscr{T}(v,\mathscr{H}^c)}$, where the latter is a closed tangent cone of $v$. If $0$ is extremal in $\mathscr{T}(v,\mathscr{C})\cap \overline{\mathscr{T}(v,\mathscr{H}^c)}$, it can not be written as sum of two vectors different from $0$, thus $0$ is also extremal in $\overline{\mathscr{T}(v,\mathscr{C}\cap \mathscr{H}^c)}$, meaning that $v$ is extremal in the closure of $\mathscr{C}\cap \mathscr{H}^c$. We have thus designed an algorithm that can recognize part of the vectors $v \in V$ that is an extremal point. For the rest of them, we can only use the brutal force to decide whether it is an extremal point.

\section{Conclusion}\label{sec:conclu}
In this work, we addressed the challenging problem of computing backward reachable sets for max-plus linear systems from general non-convex target sets. We proposed a new framework that approximates such sets using tropical polyhedra, leveraging their algebraic and geometric structure to reduce computational complexity. By reformulating the problem in terms of set operations and closure approximations, we avoided the exponential complexity of traditional DBM-based methods. Furthermore, we developed techniques for generating inner and outer representations of tropical polyhedra, and applied extremal point filtering to improve efficiency. Our approach supports a broader class of target sets beyond traditional polyhedral constraints and offers a practical tool for safety verification and control synthesis in max-plus systems. Future work may explore tighter approximation bounds and integration with controller design frameworks.

\bibliographystyle{ieeetr}
\bibliography{cite.bib}
\end{document}